\documentclass{llncs}

\usepackage{bbm}
\usepackage{amsmath}
\usepackage{alltt, amssymb}
\usepackage{mathrsfs}
\usepackage{bm}

\def\Q {\ensuremath{\mathbb{Q}}}
\def\N {\ensuremath{\mathbb{N}}}

\def\K {\ensuremath{\mathbb{K}}}
\def\L {\ensuremath{\mathbb{L}}}

\def\D {\ensuremath{D}}
\def\m {\ensuremath{\mathfrak{m}}}

\def\vm {\ensuremath{\mathsf{m}}}
\def\vv {\ensuremath{\mathsf{v}}}
\def\vw {\ensuremath{\mathsf{w}}}

\def\mM {\ensuremath{\mathsf{M}}}

\def\scrS {\ensuremath{\mathscr{S}}}
\def\ann {\ensuremath{\mathrm{ann}}}
\def\ker {\ensuremath{\mathrm{ker}}}

\DeclareBoldMathCommand{\bs}{s}
\DeclareBoldMathCommand{\bu}{u}
\DeclareBoldMathCommand{\bv}{v}
\DeclareBoldMathCommand{\bX}{X}
\DeclareBoldMathCommand{\bx}{x}
\DeclareBoldMathCommand{\balpha}{\alpha}
\DeclareBoldMathCommand{\bbeta}{\beta}

\newtheorem{Example}{Example}

\newtheorem{Prop}{Proposition}
\newtheorem{Lemma}{Lemma}

\title{Algorithms for zero-dimensional ideals using linear recurrent sequences}

\author{Vincent Neiger \inst{1}, Hamid Rahkooy\inst{2}, \and \'Eric Schost \inst{2}}
\institute{
Department of Applied Mathematics and Computer Science, Technical University of Denmark.
\and
Cheriton School of Computer Science, University of Waterloo, Canada.
}

\begin{document}

\maketitle

\begin{abstract}
Inspired by Faug\`ere and Mou's sparse FGLM algorithm, we show how
using linear recurrent multi-dimensional sequences can allow one to
perform operations such as the primary decomposition of an ideal,
by computing the annihilator of one or several such
sequences.
\end{abstract}


\section{Introduction}

In what follows, $\K$ is a perfect field. We consider the set
$\scrS=\K^{\N^n}$ of $n$-dimensional sequences $\bu = (u_m)_{m \in
  \N^n}$, and the polynomial ring $\K[X_1,\dots,X_n]$, and we are
interested in the following question. Let $I \subset
\K[X_1,\dots,X_n]$ be a zero-dimensional ideal. Given a monomial basis
of $Q=\K[X_1,\dots,X_n]/I$, together with the corresponding
multiplication matrices $\mM_1,\dots,\mM_n$, we want to compute the
Gr\"obner bases, for a target order $>$, of pairwise coprime ideals
$J_1,\dots,J_K$ such that $I = \cap_{1 \le k \le K} J_k$.

Faug\`ere {\it et al.}'s paper~\cite{FaGiLaMo93} shows how to solve
this question with $K=1$ (so $J_1$ is simply $I$) in time $O(nD^3)$,
where $D=\deg(I)$; here, the degree $\deg(I)$ is the $\K$-vector space
dimension of $Q$. More recently, algorithms have
been given with the cost bound
$O\tilde{~}(nD^\omega)$~\cite{FaGaHuRe13,FaGaHuRe14,Neiger16}, where
the notation $O\tilde{~}$ hides polylogarithmic factors, still with $K=1$. The
algorithms in this paper allow splittings (so $K > 1$ in general) and
assume that $>$ is a lexicographic order.

To motivate our approach, assume that the algebraic set $V(I)$ is in
{\em shape position}, that is, the coordinate $X_n$ separates the
points of $V(I)$. Then, the Shape Lemma~\cite{GiMo89} implies that the
Gr\"obner basis of the radical $\sqrt{I}$ for the lexicographic order
$X_1 > \cdots > X_n$ has the form $\langle
X_1-G_1(X_n),\dots,X_{n-1}-G_{n-1}(X_n),P(X_n)\rangle$, for some
squarefree polynomial $P$, and some $G_1,\dots,G_{n-1}$ of degrees
less than $\deg(P)$. The polynomials $P$ and $G_1,\dots,G_{n-1}$ can
be deduced from the values $(\ell(X_n^i))_{0 \le i \le 2D}$ and
$(\ell(X_j X_n^i))_{1 \le j < n, 0 \le i < D}$, for a randomly chosen
linear form $\ell: Q \to \K$, in time
$O\tilde{~}(D)$~\cite{BoSaSc03}. The algorithms in the latter
reference use baby steps / giant steps techniques for the calculation
of the values of $\ell$.

Similar ideas were developed in~\cite{FaMo17}; the algorithms in this
reference make no assumption on $I$ but may fail in some cases, then
falling back on the FGLM algorithm. For instance, if $I$ itself
(rather than $\sqrt{I}$) is known to have a lexicographic Gr\"obner
basis of the form $\langle
X_1-H_1(X_n),\dots,X_{n-1}-H_{n-1}(X_n),Q(X_n)\rangle$, the algorithms
in \cite{FaMo17} recover this basis, also by considering values of
linear forms $\ell_i:Q \to \K$. A key remark made in
that reference is that the values of the linear forms $\ell_i$ that we need can be
computed efficiently by exploiting the sparsity of the multiplication
matrices $\mM_1,\dots,\mM_n$; this sparsity is then analyzed, assuming
the validity of a conjecture due to
Moreno-Soc\'ias~\cite{MorenoSocias91}.
These techniques are related as well to Rouillier's Rational
Univariate Representation algorithm~\cite{Rouillier99}, which uses
values of a specific linear form $Q \to \K$ called the {\em trace}.
However, computing the trace (that is, its values on the monomial
basis of $Q$) is non-trivial, and using random choices instead makes
it possible to avoid this issue.

In this paper, we work in the continuation
of~\cite{BoSaSc03}. Assuming $V(I)$ is in shape position, the results
in that reference allow us to compute the Gr\"obner basis of
$\sqrt{I}$, and our goal here is to recover Gr\"obner
bases corresponding to a decomposition of $I$ as 
stated above. Following~\cite{FaMo17,BeBoFa16}, we discuss the relation of
this question to instances of the following problem: given sequences
$\bu_1,\dots,\bu_s$ in $\scrS$, find the Gr\"obner basis of their {\em
  annihilator} $\ann(\bu_1,\dots,\bu_s)\subset \K[X_1,\dots,X_n]$, for
a target order $>$. The annihilator, discussed in the next section, is a
polynomial ideal corresponding to the linear relations which
annihilate all sequences.

A direct approach to solve the FGLM problem using such techniques
would be to pick initial conditions at random; knowing multiplication
matrices modulo $I$ allows us to compute the values of a sequence
$\bu$, for which $I$ is contained in $\ann(\bu)$. If $I = \ann(\bu)$
holds, computing sufficiently many values of $\bu$ and feeding them
into an algorithm such as Sakata's~\cite{Sakata90} would solve our
problem. This is often, but not always, possible: there exists a
sequence $\bu$ for which $I = \ann(\bu)$ if and only if
$Q=\K[X_1,\dots,X_n]/I$ is a {\em Gorenstein} ring, a notion going back
to~\cite{Macaulay34,Groebner35} (see e.g.~\cite[Prop.\,5.3]{BrCoMoTs10} for a
proof of the above assertion). This is for instance the case if $I$ is a
complete intersection, or if $I$ is radical over a perfect field \cite{eisenbud};
however, an ideal such as $I=\langle X_1^2,X_1X_2,X_2^2\rangle \subset
\K[X_1,X_2]$ is not Gorenstein.

To remedy this, we may have to use more than one sequence, so as to be
able to recover $I$ as $I=\ann(\bu_1,\dots,\bu_s)$. However,
proceeding directly in this manner, we do not expect the algorithm to
be significantly better than applying directly the FGLM algorithm (the
techniques we will use for computing annihilators follow essentially
the same lines as the FGLM algorithm itself). We will see that
starting from the Gr\"obner basis of $\sqrt{I}$, we will be able to
decompose $I$ into e.g. primary components (assuming we allow the use
of factorization algorithms over $\K$), and that our approach is
expected to be competitive in those cases where the multiple
components of $I$ have low degrees.

\smallskip\noindent{\rm\bf Acknowledgements.} We thank the reviewers for
their remarks and suggestions. The third author is supported by an NSERC
Discovery Grant.

\section{Generalities on sequences and their annihilators}\label{ssec:linforms}



Define the shift operators $s_1,\dots,s_n$ on $\scrS$ in the obvious manner, by
setting $s_i (\bu) = (u_{m + e_i})_{m \in \N^n}$, where $e_1,\dots,e_n$ are the
unit vectors. This  makes $\scrS$ a $\K[X_1,\dots,X_n]$-module, by setting $f
\cdot \bu = f(s_1,\dots,s_n)(\bu)$. For $f= \sum_m f_m \bX^m$, the
entries of $f \cdot \bu$ are thus $( \langle \bu \mid \bX^m f \rangle)_{m \in \N^n}$,
where we write $\bX^m=X_1^{m_1} \cdots X_n^{m_n}$ and $\langle \bu \mid f \rangle
= \sum_{m'} f_{m'} u_{m'} $.
To a sequence $\bu = (u_m)_{m \in \N^n}$ in $\scrS$, we can then associate its
{\em annihilator} $\ann(\bu)$, defined as the ideal of all polynomials $f$ in
$\K[X_1,\dots,X_n]$ such that $f\cdot \bu=0$.  If we consider several sequences
$\bu_1,\dots,\bu_s$ in $\scrS$, we then define $\ann(\bu_1,\dots,\bu_s) =
\ann(\bu_1) \cap \cdots \cap \ann(\bu_s)$.

We will also occasionally discuss {\em kernels} of sequences. For
$\bu\in\scrS$, the kernel ${\rm ker}(\bu)$ is the $\K$-vector space formed by
all polynomials $f$ in $\K[X_1,\dots,X_n]$ such that $\langle \bu \mid
f\rangle=0$; this is not an ideal in general. If we consider several sequences
$\bu_1,\dots,\bu_s$, we will write $\ker(\bu_1,\dots,\bu_s) = \ker(\bu_1) \cap
\cdots \cap \ker(\bu_s)$.



Let $I$ be a zero-dimensional ideal in $\K[X_1,\dots,X_n]$.  Define
the residue class ring $Q=\K[X_1,\dots,X_n]/I$ and let
$\D=\deg(I)=\dim_\K(Q)$.  Consider also the dual $Q^*={\rm hom}_\K (Q,\K)$. To a
linear form $\ell$ in $Q^*$, we associate the sequence
$\bu_\ell$ defined by $\bu_\ell=(\ell(\bX^{m}
\bmod I))_{m \in \N^n}$.

For any linear form $\ell$ on $Q$, and any $g$ in $Q$, define the
linear form $g \cdot \ell\in Q^*$ by $(g \cdot \ell) (h) = \ell(gh)$.
This induces a $Q$-module structure on $Q^*$, and we remark that we
have the equality $g \cdot \bu_\ell =\bu_{(g \bmod I)\cdot \ell}$ for
any $g$ in $\K[X_1,\dots,X_n]$.  Following~\cite{Shoup95} (where it is
described with $n=1$), we call this operation {\em
  transposed product}.

For $\ell$ in $Q^*$, we can then define $\ann_Q(\ell)$ as the set of all
$g$ in $Q$ such that $g \cdot \ell = 0$; this is an ideal of $Q$. The
following lemma clarifies the relation between $\ann(\bu_\ell) \subset
\K[X_1,\dots,X_n]$ and $\ann_Q(\ell) \subset Q$; it implies that
$\ann(\bu_\ell)$ is generated by $I$ and any element of $\ann_Q(\ell)$ lifted to
$\K[X_1,\dots,X_n]$.

\begin{Lemma}\label{lemma:inclusion}
  With notation as above, for $f$ in $\K[X_1,\dots,X_n]$, $f$ is in
  $\ann(\bu_\ell)$ if and only if $f \bmod I$ is in $\ann_Q(\ell)$.
\end{Lemma}
\begin{proof}
  Take $f$ in $\K[X_1,\dots,X_n]$. Then $f$ is in $\ann(\bu_\ell)$ if
  and only if $f \cdot \bu_\ell = 0$, that is, if and only if $\bu_{(f
    \bmod I)\cdot \ell} = 0$, if and only if $(f \bmod
  I)\cdot \ell$ itself is zero.\qed
\end{proof}

When $Q^*$ is a free $Q$-module of rank one, we say that $Q$ is a
Gorenstein ring, and that $I$ is Gorenstein. In this case, there
exists a linear form $\lambda$ such that $Q^*= Q \cdot \lambda$; by
the previous lemma, $\ann(\bu_\lambda)=I$.  Conversely, if
$\ann(\bu_\lambda)=I$, $\ann_Q(\lambda) = \{0\}$, so that $Q^*= Q
\cdot \lambda$ (and $Q^*$ is free of
rank one).  For instance, it is known that if $I$ is radical, or $I$ a
complete intersection, then $I$ is Gorenstein. On the other hand, if
$I = \langle X_1^2, X_1 X_2, X_2^2\rangle$, the inclusion $I \subset
\ann(\bu_\ell)$ is strict for any linear form $\ell$. Using several
sequences, we can however always recover $I$.

\begin{Lemma}\label{lemma:Dgen}
Let $\ell_1,\dots,\ell_\D$ be linearly independent in $Q^*$, and let $\bu_1,\dots,\bu_\D$ be the corresponding
sequences. Then $\ann(\bu_1, \dots,\bu_\D)=\ker(\bu_1,\dots,\bu_D)=I$.
\end{Lemma}
\begin{proof}
   Note first that the inclusion $I \subset\ann(\bu_1, \dots,\bu_\D)= \ann(\bu_1) \cap \cdots
   \cap \ann(\bu_\D)$ is a direct consequence of
   Lemma~\ref{lemma:inclusion}, and that $\ann(\bu_1, \dots,\bu_\D)$ 
   is contained in $\ker(\bu_1, \dots,\bu_\D)$.
   For the converse, let $\omega_1,\dots,\omega_\D$ be the basis of
   $Q$ dual to $\ell_1,\dots,\ell_\D$. Suppose that $f$ is in
   $\ker(\bu_1, \dots,\bu_\D)$, and assume without loss of generality
   that $f$ has been reduced by $I$, so that $f$ is a linear
   combination of the form $f_1 \omega_1 + \cdots + f_\D
   \omega_\D$. Fix $i$ in $1,\dots,\D$ and apply $\ell_i$ to $f$; we
   obtain $f_i$. On the other hand, because $f$ is in $\ker(\bu_i)$, $\ell_i(f)$ must vanish. So we are done.\qed
\end{proof}

We may however need less than $D$ linear forms, as explained in the
following discussion, which generalizes the comments we made in the
Gorenstein case.

Let $B=(b_1,\dots,b_D)$ be a monomial basis of $Q$. Given a linear
form $\ell$ in $Q^*$, we define $K_{\ell}$ as the $D\times D$ matrix
whose $(i,j)$th entry is $\ell(b_i b_j)$; this is the matrix of the
mapping $f \in Q \mapsto f\cdot \ell \in Q^*$, so that its nullspace
is $\ann_Q(\ell)$. More generally, given a positive integer $s$ and
linear forms $\ell_1,\dots,\ell_s$, we define
$K_{\ell_1,\dots,\ell_s}$ as the $D \times sD$ matrix obtained as the
concatenation of $K_{\ell_1},\dots,K_{\ell_s}$; this is the matrix of
the mapping $(f_1,\dots,f_s) \in Q^s \mapsto f_1 \cdot \ell_1 + \cdots
+ f_s \cdot \ell_s \in Q^*$.

\begin{Lemma}\label{lemma:genmod}
 For any linear forms $\ell_1,\dots,\ell_s$ all in $Q^*$, $\ann(\bu_{\ell_1},
 \dots, \bu_{\ell_s}) = I$ if and only if $(\ell_1,\dots,\ell_s)$ are
 $Q$-module generators of $Q^*$.
\end{Lemma}
\begin{proof}
   $(\ell_1,\dots,\ell_s)$ are $Q$-module generators of $Q^*$ if and
  only if $K_{\ell_1,\dots,\ell_s}$ has rank $D$, if and only if
  $K_{\ell_1,\dots,\ell_s}^\perp$ has a trivial nullspace.  The nullspace
  of this matrix is the intersection of those of the matrices
  $K_{\ell_1}^\perp,\dots,K_{\ell_s}^\perp$. All these matrices are symmetric,
  and we saw that for all $i$, the nullspace of
  $K_{\ell_i}^\perp=K_{\ell_i}$ is $\ann_Q(\ell_i)$; thus, the condition above is
  equivalent to $\ann_Q(\ell_1) \cap \cdots \cap \ann_Q(\ell_{s})=\{0\}$.
  Lemma~\ref{lemma:inclusion} shows that this is the case if and only 
  if $\ann(\bu_{\ell_1}) \cap \cdots \cap \ann(\bu_{\ell_s}) = I$.\qed
\end{proof}

\begin{Prop}\label{prop:tau}
  There exists a unique integer $\tau \le D$ such that for a generic
  choice of linear forms $(\ell_1,\dots,\ell_\tau)$, with all $\ell_i$
  in $Q^*$, the sequence of ideals $(\ann(\bu_{\ell_1}, \dots,
  \bu_{\ell_t}))_{1 \le t \le \tau}$ is strictly decreasing,
  with $\ann(\bu_{\ell_1}, \dots, \bu_{\ell_\tau}) =
  I$.
\end{Prop}
\begin{proof}
  Remark first that if $\tau$ exists with the properties above, it 
  is necessarily unique. Let $(L_{1,1},\dots,L_{1,D}),\dots,(L_{D,1},\dots,L_{D,D})$ be new
  indeterminates, let $\L=\K(L_{1,1},\dots,L_{D,D})$ and define the
  matrices $K_{L_1},\dots,K_{L_D}$ as follows. Let $Q_\L=Q \otimes_\K
  \L$; this allows us to define the linear forms $L_1,\dots,L_D$ in
  $Q_\L^*$ by $L_t(b_j) = L_{t,j}$, for $1 \le t \le D$; then
  $K_{L_t}$ is the matrix with entries $L_t(b_i b_j)$. The entries
  of $K_{L_t}$ are linear forms in $L_{t,1},\dots,L_{t,D}$.

  Define $K_{L_1,\dots,L_t}$ as we did for $K_{\ell_1,\dots,\ell_t}$.
  Then, for any linear forms $\ell_1,\dots,\ell_t$ in $Q^*$, the matrix
  $K_{\ell_1,\dots,\ell_t}$ is obtained by evaluating
  $K_{L_1,\dots,L_t}$ at $L_{t,j} = \ell_t(b_j)$, for all $t,j$.
  The rank of $K_{\ell_1,\dots,\ell_t}$ (over $\K$) is at most that of 
  $K_{L_1,\dots,L_t}$ (over $\L$).

  We can then let $\tau$ be the smallest integer such that the matrix
  $K_{L_1,\dots,L_\tau}$ has full rank $D$. Such an index exists, and
  is at most $D$, since by Lemma~\ref{lemma:Dgen} (and by the remarks of 
  the above paragraph) $K_{L_1,\dots,L_D}$ has rank $D$.

  Let $\ell_1,\dots,\ell_\tau$ be such that
  $K_{\ell_1,\dots,\ell_\tau}$ has rank $D$ (this is our genericity
  condition); in this case, by the previous lemma,
  $\ann(\bu_{\ell_1},\dots,\bu_{\ell_\tau}) = I$. To conclude, it
  suffices to prove that the sequence of ideals $(\ann(\bu_{\ell_1},
 \cdots,\bu_{\ell_t}))_{1 \le t \le \tau}$ is strictly
  decreasing.   Suppose it is not the case, so that $\ann(\bu_{\ell_1}, \dots
,\bu_{\ell_t}) = \ann(\bu_{\ell_1}, \dots,\bu_{\ell_{t+1}})$ for some $t < \tau$. Then, 
$\ann(\bu_{\ell_1}, \dots,\bu_{\ell_t},\bu_{\ell_{t+2}}, \dots,\bu_{\ell_\tau}) =I.$
  Let us define
  $\ell'_1=\ell_1,\dots,\ell'_t=\ell_t,\ell'_{t+1}=\ell_{t+2},\dots,\ell'_{\tau-1}=\ell_\tau$.
  Then, we have $\ann(\bu_{\ell'_1}, \dots,\bu_{\ell'_{\tau-1}}) = I$, so that 
  $K_{\ell'_1,\dots,\ell'_{\tau-1}}$ has rank $D$. This in turn
  implies (by the discussion above) that 
  $K_{L_1,\dots,L_{\tau-1}}$ has rank $D$, a contradiction.\qed
\end{proof}

If $Q$ is a local algebra with maximal ideal $\m$, we can define the
{\em socle} of $Q$ as the $\K$-vector space of all elements $f$ in $Q$
such that $\m f=0$.  For instance,   if $Q$ is local, the integer $\tau$ in the previous lemma is the
  dimension of the socle of $Q$.
(we omit the proof, since we will not use this result in the rest of the paper).


\section{Computing annihilators of sequences}

Consider sequences $(\bu_1,\dots,\bu_t)$ with $\bu_i \in \scrS$ for
all $i$, let $J$ be the annihilator $\ann(\bu_1,\dots,\bu_t)
\subset \K[X_1,\dots,X_n]$, and suppose that it has dimension
zero; our goal is to compute a Gr\"obner basis of it. We first review an
algorithm due to Marinari, M\"oller and Mora~\cite{MaMoMo93}, then
introduce a modification of it that relaxes some of its
assumptions. As a result, the algorithms in this section work under
slightly different assumptions, and feature slightly different
runtimes.

An algorithm with cost $(n t \deg(J))^{O(1)}$ would be highly
desirable, but we are not aware of any such result. Most approaches
(ours as well) involve reading a number of values of $\bu_1,\dots,\bu_t$ and looking
for dependencies between the columns of what is often called a
generalized Hankel matrix, built using these values; the delicate
question is how to control the size of the matrix.

Consider for instance the case $t=1$, $\langle \bu_1 \mid X_1^{m_1} \cdots X_n^{m_n}
\rangle=1$ for $m_1+\cdots+m_n < \delta$ and $\langle \bu_1 \mid X_1^{m_1} \cdots X_n^{m_n} \rangle=0$
otherwise.  The annihilator $J=\ann(\bu_1)$ admits the
lexicographic Gr\"obner basis
$\langle X_1-X_n,\dots,X_{n-1}-X_n,X_n^\delta\rangle$, so we have $\deg(J)=\delta$;
on the other hand, this sequence takes ${\deg(J)+n-1} \choose n$
non-zero values, so taking them all into account leads us to an
exponential time algorithm.


In the case $t=1$, Mourrain in~\cite{mourrain-goren-alg} associates a Hankel operator 
to a sequence
such that the kernel of the Hankel operator corresponds to the annihilator of 
the sequence. Algorithm 2 in this paper computes a border basis 
for the kernel of such a Hankel operator, taking as input its values
over a finite set of monomials.
As in the FGLM algorithm, this algorithm looks for linear dependencies between
the monomials in the border of already computed linearly independent
monomials.  
However, for examples as in the previous paragraph, we
are not aware of how to avoid taking into account up to ${\deg(J)+n-1}
\choose n$ values.


Several algorithms were also proposed in \cite{BeBoFa16} for computing
an annihilator $\ann(\bu_1)$, and partly extended to arbitrary $t$
in~\cite{BeFa16}.  A first algorithm relies on the Berlekamp-Massey
Algorithm, by means of a change of coordinates, which may require\ an
exponential number of value of $\bu_1$. The other algorithms extend
the idea of FGLM, considering maximal rank sub-matrices of a
truncated multi-Hankel matrix to compute a basis for the quotient algebra
and a Gr\"obner basis. An algorithm with certified outcome
(Scalar-FGLM) is presented; it considers the values of $\bu_1$ at all
monomials up to a given degree $\simeq\deg(J)$, so the issue pointed out
above remains. An ``adaptive'' version uses fewer values of the
sequence, but may fail in some cases (the conditions that ensure 
success of this algorithm seem to be close to the genericity assumptions
we introduce in Subsection~\ref{ssec:algo2}).
A comparison of Scalar-FGLM and Sakata's algorithm 
is presented in~\cite{beFa17}.



\subsection{A first algorithm}\label{ssec:algo1}

The first solution we discuss requires a strong assumption (written
${\sf H}_1$ below): for any $i$ and for any monomial $b$ in
$X_1,\dots,X_n$, $b \cdot \bu_i$ is in the $\K$-span of
$(\bu_1,\dots,\bu_t)$; as a result, the annihilator $J$ of
$(\bu_1,\dots,\bu_t)$ equals the nullspace $\ker(\bu_1,\dots,\bu_t)$.
For this situation, Marinari, M\"oller and Mora gave
in~\cite{MaMoMo93} an algorithm that compute a Gr\"obner basis of
$J$, for any order (for definiteness, we refer here to their second
algorithm); it is an extension of both the Buchberger-M\"oller
interpolation algorithm and the FGLM change of order algorithm.

Assumption ${\sf H}_1$ above implies that $\deg(J) \le t$, and the
runtime of the algorithm, expressed in terms of $n$ and $t$, is
$O(nt^3)$ operations in $\K$, together with the computation of all
values $\langle \bu_i \mid b \rangle$, $1 \le i \le t$, for $O(nt)$
monomials $b$.  These evaluations are done in incremental order, in
the sense that for any monomial $b$ for which we need all $\langle \bu_i
\mid b\rangle$, there exists $j \in \{1,\dots,n\}$ such
that $b = X_j b'$ and all $\langle \bu_i \mid b'\rangle$ are known.

We will need the following property of this algorithm. Suppose
$(\bu_1,\dots,\bu_t)$ is a subsequence of a larger family of sequences
$(\bu_1,\dots,\bu_{t'})$ that satisfies ${\sf H}_1$, but that
$(\bu_1,\dots,\bu_t)$ itself may or may not, and that $(\bu_1,\dots,\bu_t)$ and
$(\bu_1,\dots,\bu_{t'})$ have different $\K$-spans. Then, on input
$(\bu_1,\dots,\bu_t)$, the algorithm will still run its course, and at least
one of the elements in the output will be a polynomial $g$ that does not belong
to $\ann(\bu_1,\dots,\bu_{t'})$.

\subsection{An algorithm under genericity assumptions}\label{ssec:algo2}

We now give a second algorithm for computing
$J=\ann(\bu_1,\dots,\bu_t)$, whose runtime is polynomial in
$n,t,D=\deg(J)$ and an integer $B\le \deg(J) $ defined below.  We do
not assume that ${\sf H}_1$ holds, but we will require other
assumptions; if they hold, the output is the lexicographic Gr\"obner
basis $G$ of $J$ for the order $X_1 > \cdots > X_n$. Our first assumption is:
\begin{description}
\item [${\sf H}_2.$] We are given an integer $B$ such that the minimal
  polynomial of $X_j$ in $\K[X_1,\dots,X_n]/J$ has degree at most $B$
  for all $j$.
\end{description}
For $j$ in $1,\dots,n$, we will denote by $J_j$ the ideal
$\ann(\pi_j(\bu_1),\dots,\pi_j(\bu_t)) \subset \K[X_{j},\dots,X_n]$,
where for all $i$, $\pi_j(\bu_i)$ is the sequence $\N^{n-j+1}\to \K$
defined by $\langle \pi_j(\bu_i) \mid (m_j,\dots,m_n)\rangle =\langle
\bu_i \mid (0,\dots,0,m_j,\dots,m_n)\rangle$ for all $(m_j,\dots,m_n)$
in $\N^{n-j+1}$; in particular, $J_1=J$. We write $\deg(J_j)=D_j\le D$, we
let $G_{j}$ be the lexicographic Gr\"obner basis of $J_j$,
and we let $\mathscr{B}_j$ be the corresponding monomial basis of
$\K[X_j,\dots,X_n]/J_j$.

We can then introduce our genericity property; by contrast with ${\sf
  H}_2$, we will not necessarily assume that it holds, and discuss the
outcome of the algorithm when it does not. We denote this property by
${\sf H}_3(j)$, for $j=1,\dots,n-1$.
\begin{description}
\item [${\sf H}_3(j).$] We have the equality $J_j \cap \K[X_{j+1},\dots,X_n]=J_{j+1}$.
\end{description}
Remark that the inclusion $J_j \cap \K[X_{j+1},\dots,X_n] \subset
J_{j+1}$ always holds.

Suppose that for some $j$ in $1,\dots,n$, we have computed a sequence
of monomials $\mathscr{B}'_{j+1}$ in $\K[X_{j+1},\dots,X_n]$ (if
$j=n$, we let $\mathscr{B}'_{j+1}=(1)$). Since we will use them
repeatedly, we define properties ${\sf P}$ and ${\sf P}'$ as follows,
the latter being stronger than the former.
\begin{description}
\item [${\sf P}(j+1).$] The cardinality $D'_{j+1}$ of $\mathscr{B}'_{j+1}$ is at
  most $D_{j+1}$.
\item [${\sf P}'(j+1).$] The equality
  $\mathscr{B}'_{j+1}=\mathscr{B}_{j+1}$ holds.
\end{description}
We describe in the following paragraphs a procedure that computes a
new family of monomials $\mathscr{B}'_j$, and we give conditions under
which they satisfy ${\sf P}(j)$ and ${\sf P}'(j)$.

We call a family of monomials $\mathscr{B}$ in $\K[X_j,\dots,X_n]$
{\em independent} if their images are $\K$-linearly independent modulo
$J_j$ (we call it {\em dependent} otherwise). We denote by
$\mM_{\mathscr{B}}$ the matrix with entries $\langle \bu_i \mid b b'
\rangle$, with rows indexed by $i=1,\dots,t$ and $b'$ in
$\mathscr{C}_{j+1}=\mathscr{B}'_{j+1} \times (1,X_j,\dots,X_j^{B-1})$,
and columns indexed by $b$ in $\mathscr{B}$ (for any monomial $b$ in
$\K[X_j,\dots,X_n]$, $\mM_b$ is the column vector defined similarly).

\begin{Lemma}\label{lemma:rnkM}
  If $\mathscr{B}$ is dependent, the right nullspace of
    $\mM_{\mathscr{B}}$ is non-trivial.
  If both ${\sf P}'(j+1)$ and ${\sf H}_3(j)$ hold, the converse
    is true.
 \end{Lemma}
\begin{proof}
  Any $\K$-linear relation between the elements of $\mathscr{B}$
  induces the same relation between the columns of
  $\mM_{\mathscr{B}}$, and the first point follows.

  By definition, a polynomial $f$ in $\K[X_j,\dots,X_n]$ belongs to
  $J_j$ if and only if it annihilates
  $\pi_j(\bu_1),\dots,\pi_j(\bu_t)$, that is, if $\langle  \pi_j(\bu_i) \mid X_j^{m_j}
  \dots X_n^{m_n} f \rangle=0$ for all
  $(m_j,\dots,m_n)$ in $\N^{n-j+1}$ and all $i=1,\dots,t$. Now,
  assumptions  ${\sf P}'(j+1)$, ${\sf H}_2$ and ${\sf H}_3(j)$ imply that
  $\mathscr{C}_{j+1}$ generates $\K[X_j,\dots,X_n]/J_j$, so that $f$
  is in $J_j$ if and only if $\langle \bu_i \mid b f \rangle=0$, for all $b$ in
  $\mathscr{C}_{j+1}$ and all $i=1,\dots,t$.  \qed
\end{proof}

The following lemma, that essentially follows the argument used in the
proof of the FGLM algorithm~\cite{FaGiLaMo93}, will be useful to
justify our algorithm as well.

\begin{Lemma}\label{lemma:fglm}
  Suppose that $b_1 < \dots < b_u < b_{u+1}$ are the first $u+1$ standard
  monomials of $\K[X_j,\dots,X_n]/J_j$, for the lexicographic order 
  induced by $X_j > \cdots > X_n$, with $b_1=1$. Then for any monomial 
$b$ such that $b_u < b < b_{u+1}$, $\{b_1,\dots,b_u,b\}$ is a dependent 
  family.
\end{Lemma}
\begin{proof}
  We prove the result by induction on $u \ge 0$, the case $u=0$ being
  vacuously true. Assuming the claim is true for some index $u \ge 0$,
  we prove it for $u+1$. We proceed by contradiction, and we let $b$
  be the smallest monomial such that $b_u < b < b_{u+1}$ and
  $\{b_1,\dots,b_u,b\}$ is an independent family ($b$ exists by the
  well-ordering property of monomial orders). 

  We will use the fact that any monomial $c$ less than $b$ can be
  rewritten as a linear combination of $b_1,\dots,b_i$, with $b_i <
  c$, for some $i \le u$: if $c < b_u$, this is by the induction
  assumption; if $c=b_u$, this is obvious; if $b_u < c < b$, this is
  by the definition of $b$.

  Now, either $b$ is the leading term of an element in the Gr\"obner
  basis of $J_j$, or it must be of the form $b=X_e b'$, for some
  monomial $b'$ not in $\{b_1,\dots,b_u\}$. We prove that in both
  cases, $b$ can be rewritten as a linear combination of
  $b_1,\dots,b_u$, which is a contradiction.
  In the first case, $b$ rewrites as a linear combination of smaller
  monomials, say $c_1,\dots,c_v$, and by the previous remark, all of
  them can be rewritten as linear combinations of
  $b_1,\dots,b_u$. Altogether, $b$ itself can be rewritten as a linear
  combination of $b_1,\dots,b_u$, a contradiction.
  
  In the second case, $b=X_e b'$, for some monomial $b'$ not in
  $\{b_1,\dots,b_u\}$. As above, $b'$ can be rewritten modulo $J_j$ as
  a linear combination of monomials $b_1,\dots,b_i$, for some $i \le
  u$, with $b_i < b'$. Then, $b=X_e b'$ is a linear combination of
  $X_e b_1,\dots,X_e b_i$. Since $b_i < b'$, we get $X_e b_1 < \cdots
  < X_e b_i < X_e b'=b$, so all of $X_e b_1,\dots,X_e b_i$ can be
  rewritten as linear combinations of $b_1,\dots,b_u$. As a result,
  this is also the case for $b$ itself, so we get a contradiction again.
\qed
\end{proof}

\noindent Suppose that ${\sf P}(j+1)$ holds.
Then, the algorithm at step $j$ proceeds as follows.  We compute the
reduced row echelon form of $\mM_{\mathscr{C}_{j+1}}$. Using
assumption ${\sf P}(j+1)$, this matrix has at most $t B D_{j+1}$ rows
and at most $B D_{j+1}$ columns, and it has rank at most $D_j$ (by the
first item of Lemma~\ref{lemma:rnkM}). This computation can
be done in time $O(t B^2 D_{j+1}^2 D_j) \in O(t B^2 D^3)$. The column
indices of the pivots allow us to define the monomials $\mathscr{B}'_j=(b'_1 <
\cdots < b'_{D'_j})$, for some $D'_j \le D_j$.

\begin{Lemma} Property ${\sf P}(j)$ holds, and if ${\sf P}'(j+1)$ and ${\sf H}_3(j)$ hold, then ${\sf P}'(j)$
  holds.
\end{Lemma}
\begin{proof}
  The first item is a restatement of the inequality $D'_j \le D_j$.
  To prove the second item, assuming that ${\sf P}'(j+1)$ and ${\sf
    H}_3(j)$ hold, we deduce from Lemma~\ref{lemma:rnkM} that the
  columns indexed by the genuine $\mathscr{B}_{j}$ form a column basis
  of $\mM_{\mathscr{C}_{j+1}}$, and we claim that it is actually the
  lexicographically smallest column basis (this will prove that 
  $\mathscr{B}_{j}=\mathscr{B}'_{j}$).
  Indeed, write $\mathscr{B}_{j}=(b_1,\dots,b_{D_j})$, and let
  $(f_1,\dots,f_{D_j})$ be another subsequence of $\mathscr{C}_{j+1}$
  whose corresponding columns form a column basis of
  $\mM_{\mathscr{C}_{j+1}}$.  Let $m$ be the smallest index such that
  $b_{m}\ne f_{m}$. Then, applying Lemma~\ref{lemma:fglm} to
  $(b_1,\dots,b_{m-1})$ and $f_{m}$, we deduce that $b_{m} < f_{m}$
  (otherwise, since they are different, we must have $b_{m-1} < f_{m}
  < b_{m}$, which implies that $f_{m}$ is a linear combination of
  $(b_1,\dots,b_{m-1})=(f_1,\dots,f_{m-1})$, a contradiction).  \qed
\end{proof}

Thus, running this procedure for $j=n,\dots,1$, we maintain ${\sf
  P}(j)$; this implies that the running time is $O(n t B^2 D^3)$,
computing the values $\langle \bu_i \mid b\rangle$, for $1 \le i \le
t$, for $O(nB^2D^2)$ monomials $b$ (with the same monotonic property
as in the previous subsection). If ${\sf H}_3(j)$ holds for all $j$,
the second item in the last lemma proves that
$\mathscr{B}'_1=\mathscr{B}_1$, the monomial basis of
$\K[X_1,\dots,X_n]/J$.

Once $\mathscr{B}'_{1}$ is known, we compute and return a family of polynomials
$G'$ defined as follows. We determine the sequence $\Delta$ of elements in
$X_1 \mathscr{B}'_1 \cup \dots \cup X_n
\mathscr{B}'_1-\mathscr{B}'_1$, all of whose factors are in
$\mathscr{B}'_1$ (finding them does not require any operation in $\K$;
this can be done by using e.g. a balanced binary search tree with the
elements of $\mathscr{B}'_1$, using a number of comparisons that is
quasi-linear time in $n D$). Then, we rewrite each column $\mM_{b}$,
for $b$ in $\Delta$, as a linear combination of the form $\sum_{1 \le
  i \le D'_1} c_i \mM_{b'_i}$ and we put $b - \sum_{1 \le i \le D'_1}
c_i b'_i$ in $G'$. If the reduction is not possible, the algorithm
halts and returns {\sf fail}.  Using the reduced row echelon form of
$\mM_{\mathscr{C}_2}$, each reduction takes time $O(D_1^2)\in O(D^2)$ operations
in $\K$, for a total of $ O(nD^3)$.

If ${\sf H}_3(j)$ holds for all $j$, since
$\mathscr{B}_1=\mathscr{B}'_1$, the fact that $G'=G$ follows from
Lemma~\ref{lemma:rnkM}. 
Assume now that $G'$ differs from $G$; we prove that there exists an
element in $G$ not in $J$ (we will use this in our main algorithm to
detect failure cases). Indeed, in this case, $\mathscr{B}'_1$ must be
different from $\mathscr{B}_1$, and since $\mathscr{B}'_1$ has
cardinality at most equal to that of $\mathscr{B}_1$, there exists a
monomial $b$ in $\mathscr{B}_1$ not in $\mathscr{B}'_1$. This in turn
implies that there exists an element $g$ in $G'$ that divides $b$, and
thus with leading term in $\mathscr{B}_1$. Reducing $g$ modulo $G$, we
must then obtain a non-zero remainder, so that $g$ does not belong to
$J$.


\section{Main algorithm}


\subsection{Representing primary zero-dimensional ideals}\label{sec:output}

Let $I$ be a zero-dimensional ideal in $\K[X_1,\dots,X_n]$; we assume
that $I$ is $\m$-primary, for some maximal ideal $\m$, and we write
$D=\deg(I)$. In this paragraph, we briefly mention some possible
representations for $I$ (our main algorithm will compute either one of
these representations).

The first, and main, option we will consider is simply the Gr\"obner
basis $G$ of $I$, for the lexicographic order induced by $X_1 > \cdots
> X_n$.
As an alternative, consider the following construction. Our
assumption on $I$ implies that the minimal polynomial $R$ of $X_n$ in
$\K[X_1,\dots,X_n]/I$ takes the form $R=P^e$, for some irreducible
polynomial $P$ in $\K[Z]$, of degree say $f$ (remark that $R(X_n)$ is
also the last polynomial in $G$). Let $\L=\K[Z]/\langle P\rangle$;
this is a field extension of degree $f$ of $\K$, and the residue class
$\zeta$ of $Z$ in $\L$ is a root of $P$. We then let $I'$ be the ideal
$I + \langle (X_n-\zeta)^{e}\rangle$ in $\L[X_1,\dots,X_n]$, and let
$D'$ be its degree. Then, a second option is to compute the
lexicographic Gr\"obner basis $G'$ of $I'$, for the order $X_1 >
\cdots > X_n$. The following lemma relates $D$ and $D'$.

\begin{Lemma}\label{lemma:DDprime}
  The ideal $I'$ has degree $D'=D / f$. 
\end{Lemma}
\begin{proof}
  Let $\mathbb{M}$ be the splitting field of $P$ and let
  $\zeta_1,\dots,\zeta_f$ be the roots of $P$ in $\mathbb{M}$.  The ideals
  $J_i=I + \langle (X_n-\zeta_i)^{e}\rangle \subset
  \mathbb{}[X_1,\dots,X_n]$ are such that
  $\deg(J_1)+\cdots+\deg(J_f)=\deg(I)$. On the other hand, there
  exist $f$ embeddings $\sigma_1,\dots,\sigma_f$ of $\L$ into $\mathbb{M}$,
  with $\sigma_i$ given by $\zeta \mapsto \zeta_i$; as a result,
  $\deg(I')=\deg(J_i)$ holds for all $i$, and the claim follows.\qed
\end{proof}

The point behind this construction is to lower the degree of the ideal
we consider, at the cost of working in a field extension of $\K$. This
may be beneficial, as the cost of the main algorithm (which
essentially relies on the one in the previous section) will be a
polynomial of rather large degree with respect to the degree of the
ideal, whereas computation in a field extension such as $\K \to \L$ is
a well-understood task of cost ranging from quasi-linear to quadratic.

Our last option aims at producing a ``simpler'' Gr\"obner basis, by
means of a change of coordinates. For this, we will assume that $X_n$
separates the points of $V(\m)$ (over an algebraic closure of
$\K$). As a result, the ideal $\m$ being maximal, it admits a
lexicographic Gr\"obner basis of the form $\langle
X_1-G_1(X_n),\dots,X_{n-1}-G_{n-1}(X_n),P(X_n)\rangle$. Define
$\xi_1=G_1(\zeta),\dots,\xi_{n-1}=G_{n-1}(\zeta),\xi_n=\zeta$, for
$\zeta \in \L$ as above; then, $(\xi_1,\dots,\xi_n)$ is the unique
zero of $I'$ (in fact, $I'$ is $\m'$-primary, with $\m'=\langle X_1
-\xi_1,\dots,X_n-\xi_n\rangle$). We can then apply the change of
coordinates that replaces $X_i$ by $X_i+\xi_i$ in $I'$, for all $i$,
and call $I''$ the ideal thus obtained (so that $I''$ is generated by
the polynomials $f(X_1+\xi_1,\dots,X_n+\xi_n)$, for $f$ in $I$, and
$X_n^e$). Now, $I''$ is $\m''$-primary, with $\m''=\langle
X_1,\dots,X_n\rangle$; one of our options will be to compute the
Gr\"obner basis $G''$ of $I''$.

\begin{Example}
  Consider the polynomials in $\Q[X_1,X_2]$
\begin{align*}
X_1^2 - 2X_1X_2 - 2X_1 + X_2^2 + 2X_2 + 1,\\
X_1X_2^2 + X_1X_2 + 2X_1 - X_2^3 - 2X_2^2 - 3X_2 - 2,\\
X_2^4 + 2X_2^3 + 5X_2^2 + 4X_2 + 4,
\end{align*}
the last of them being $P(X_2)^2=(X_2^2+X_2+2)^2$, and let $I$ be the
ideal they define. The polynomials above are the lexicographic Gr\"obner basis $G$ of $I$ 
for the order $X_1 > X_2$. Let $\L=\Q[Z]/\langle Z^2+Z+2\rangle$, and let
$\zeta$ be the image of $Z$ in $\L$; then, the ideal $I'=I + \langle
(X_2-\zeta)^2 \rangle$ in $\L[X_1,X_2]$ admits the Gr\"obner
basis $G'$
\begin{align*}
    X_1^2 - 2X_1\zeta - 2X_1 + \zeta - 1,\\
    X_1X_2 - X_1\zeta - X_2\zeta - X_2 - 2,\\
    X_2^2 - 2X_2\zeta - \zeta - 2.
\end{align*}
Here, we have $e=2$, $f=2$, $D=6$ and $D'=3$. The ideal $I$ is $\m$-primary,
where $\m$ admits the Gr\"obner basis $\langle X_1-X_2-1, X_2^2+X_2+2 \rangle$, so that
we have $(\xi_1,\xi_2)=(\zeta+1,\zeta)$, and $I'$ is $\m'$-primary, with
$\m'=\langle X_1-\xi_1,X_2-\xi_2 \rangle$. Applying the change of coordinates
$(X_1,X_2) \leftarrow (X_1+\xi_1,X_2+\xi_2)$, the resulting ideal $I''$ admits
the Gr\"obner basis $G''=\langle X_1^2, X_1X_2, X_2^2 \rangle$,
from which we can readily confirm that it is $\langle X_1,X_2\rangle$-primary.
\end{Example}



\subsection{The algorithm}

We consider a zero-dimensional ideal $I$ in $\K[X_1,\dots,X_n]$. We
assume that we know a monomial basis $B=(b_1,\dots,b_D)$ of
$Q=\K[X_1,\dots,X_n]/I$, so that we let $D=\dim_\K(Q)$, together with
the corresponding multiplication matrices $\mM_1,\dots,\mM_n$ of
respectively $X_1,\dots,X_n$. We assume that the last variable $X_n$
has been chosen generically; in particular, $X_n$ separates the points
of $V=V(I)$.

The algorithm in this section computes a decomposition of $I$ into
primary components $J_1,\dots,J_K$. Each such component $J_k$ will be
given by means of one of the representations described in the previous
subsection; we will emphasize the first of them, the lexicographic
Gr\"obner basis of $J_k$, and mention how to modify the algorithm in
order to obtain the other representations. In order to find the
primary components of $I$, we cannot avoid the use of factorization
algorithms over $\K$; if desired, one may avoid this by relying on
{\em dynamic evaluation techniques}~\cite{D5}, replacing for instance
the factorization into irreducibles used below by a squarefree
factorization (thus producing a decomposition of $I$ into ideals that
are not necessarily primary).  In that case, if one wishes to compute
descriptions such as the second or third ones introduced above,
involving algebraic numbers as coefficients, one should take into
account the possibility of splittings the defining polynomials, as is usual
with this kind of approach (a
complete description of the resulting algorithm, along the lines
of~\cite{DaMoScXi06}, is beyond the scope of this paper).

\smallskip\noindent{\rm\bf The ideal $I$ and its primary decomposition.}
Let $P_{\min{}} \in \K[X_n]$ be the minimal polynomial of $X_n$ in
$Q$, let $P$ be its squarefree part, and let polynomials
$G_1,\dots,G_{n-1}$ in $\K[X_n]$, with $\deg(G_i) < \deg(P)$ for all
$i$, be such that $\sqrt{I}$ admits the lexicographic Gr\"obner basis
$ \langle X_1-G_1(X_n),\dots,X_{n-1}-G_{n-1}(X_n),P(X_n)\rangle$. We write
$P_{\min{}}=P_1^{e_1} \cdots P_K^{e_K}$, with the $P_k$'s pairwise
distinct irreducible polynomials in $\K[X_n]$ and $e_k \ge 1$ for
all $k$. In particular, the factorization of $P$ is $P_1 \cdots P_K$; we
write $f_k = \deg(P_k)$ for all $k$.

Correspondingly, let $V_1,\dots,V_K$ be the $\K$-irreducible
components of $V$ and for $k=1,\dots,K$, let $\m_k$ be the maximal
ideal defining $V_k$; hence, the reduced lexicographic Gr\"obner basis
of $\m_k$ is $\langle X_1- (G_1 \bmod P_k),\dots,X_{n-1}-(G_{n-1}
\bmod P_k),P_k \rangle$.  We can then write $I = J_1 \cap \cdots \cap
J_K$, with $J_k$ $\m_k$-primary for all $k$; note that
the ideal $J_k$ is defined by $J_k = I + \langle P_k^{e_k}
\rangle$. In what follows, we explain how to compute a Gr\"obner basis
of this ideal by means of the results of the previous section.
Without loss of generality, assume that $L$ is such that $e_k = 1$ for
$k > L$ and $e_k \ge 2$ for $k=1,\dots,L$. The fact that $X_n$ is a
generic coordinate implies that for $k > L$, $J_k=\m_k$, so there is
nothing left to do for such indices; hence, we are left with showing
how to use the algorithms of the previous section to compute Gr\"obner
bases of $J_1,\dots,J_L$. 

\smallskip\noindent{\rm\bf Data representation.} 
An element $f$ of $Q$ is represented by the column vector $\vv_f$ of
its coordinates on the basis $B$, whereas a linear form $\ell:Q\to\K$
is represented by the row vector
$\vw_\ell=[\ell(b_1),\dots,\ell(b_D)]$. Computing $\ell(f)$ is then
done by means of the dot product $\vw_\ell \cdot \vv_f$. Multiplying
$f$ by $X_i$ amounts to computing $\mM_i \vv_f$, and the linear form
$X_i \cdot \ell: g \mapsto \ell(X_i g)$ is obtained by computing the
vector $\vw_{X_i \cdot \ell} = \vw_\ell \mM_i$.

In terms of complexity, we assume that multiplying any matrix $\mM_i$
by a vector (either on the left or on the right) can be done in $\vm$
operations in $\K$. The naive bound on $\vm$ is $O(D^2)$, but the
sparsity properties of these matrices often result in much better
estimates; see~\cite{FaMo17} for an in-depth discussion of this question.
On the other hand, we assume $D \le \vm$.

\smallskip\noindent{\rm\bf Computing $P_{\min{}}$ and $G_1,\dots,G_{n-1}$.}
First, we compute generators of $\sqrt{I}$. We choose a random linear
form $\ell_1 : Q \to \K$, and we compute the values
$(\ell_1(X_n^i))_{0 \le i < 2D}$ and $\ell_1(X_1
X_n^i),\dots,\ell_1(X_{n-1} X_n^i)$, for ${0 \le i < D}$.  This is
done by computing $1,X_n,\dots,X_n^{2D-1}$ by repeated applications of
$\mM_n$, which amounts to $O(D \vm)$ operations, 
and doing the
corresponding dot products with $\ell,X_1 \cdot \ell,\dots,X_{n-1}
\cdot \ell$. For the latter, we have to compute the linear 
forms $X_i\cdot \ell$ in $O(n \vm)$ operations, then do a 
$D \times D$ by $D \times (n+1)$ matrix product, which costs
$O(nD^2)$ operations (without using fast linear algebra).

Using the algorithm given in~\cite{BoSaSc03}, given these values, we
can compute the minimal polynomial $P_{\min}$, as well as the
polynomials $G_1,\dots,G_{n-1}$ describing $V(I)$ in $O\tilde{~}(D)$
operations in $\K$. Then, as per the discussion in the preamble, we
assume that we have an algorithm for factoring polynomials over $\K$,
so that $(P_1,e_1),\dots,(P_K,e_K)$ and $P$ can be deduced from $P_{\min}$.

\smallskip\noindent{\rm\bf Constructing the orthogonal of $J_k$.} 
For $k=1,\dots,K$, we will write $Q_k=\K[X_1,\dots,X_n]/J_k$.
Any linear form $\ell:Q \to\K$ induces a linear form $\varphi_k(\ell):
Q_k \to \K$, defined as follows.

Let $T_k$ be the polynomial $P_{\min{}}/P_k^{e_k}$. For $f$ in $Q_k$,
let $\hat f$ be any lift of $f$ to $\K[X_1,\dots,X_n]$, and define
$\varphi_k(\ell)(f)=\ell(T_k \hat f \bmod I)$. Notice that this
expression is well-defined: indeed, any two lifts of $f$ differ by an
element $\delta$ of $J_k = I + \langle P_k^{e_k}\rangle$, so that $T_k
\delta$ is in $I$, since $T_k P_k^{e_k} = P_{\min{}}$ is.

\begin{lemma}
The mapping $\varphi_k: Q^* \to Q_k^*$ is $\K$-linear and onto. 
\end{lemma}
\begin{proof}
  Linearity is clear by construction; we now prove that $\varphi_k$ is
  onto. Let indeed $A_k,B_k$ in $\K[X_n]$ be such that $A_k T_k + B_k
P_k^{e_k}=1$ (they exist by definition of $T_k$). Consider $\lambda$
in $Q_k^*$, and define $\ell$ in $Q^*$ by $\ell(f)=\lambda(A_k f \bmod
J_k)$.  Since $ P_k^{e_k}$ vanishes modulo $J_k$, we have $A_k T_k=1
\bmod J_k$, so $\ell (f)=\lambda(f \bmod J_k)$ holds for all $f$ in
$Q$; this in turn readily implies that $\varphi_k(\ell)=\lambda$.\qed
\end{proof}

We saw in Subsection~\ref{ssec:linforms} how to associate to an
element $\ell \in Q^*$ a sequence $\bu_\ell \in \scrS$, by letting
$\langle \bu_\ell \mid m \rangle = \ell(m \bmod I)$. The following
tautological observation will then be useful below: for $\ell$ in
$Q^*$, the sequences $\bu_{T_k \cdot \ell}$ and
$\bu_{\varphi_k(\ell)}$ coincide, where $\bu_{\varphi_k(\ell)}$ is
defined starting from the linear form $\varphi_k(\ell) \in
Q_k^*$. Indeed, take any monomial $m$ in $X_1,\dots,X_n$; then,
$\varphi_k(\ell)(m \bmod J_k)$ is defined as $\ell(T_k m \bmod I)$,
which is equal to $(T_k \cdot \ell)(m \bmod I)$.
We will use this remark to compute values of
$\varphi_k(\ell)$, through the computation of values of $T_k \cdot
\ell$ instead. 

In algorithmic terms, computing a single transposed product by a
polynomial $T(X_n)$, that is, $T \cdot \ell$, can be done using
Horner's rule, using $d$ right-multiplications by~$\mM_n$, with
$d=\deg(T)$; this takes $O(d\vm)$ operations in $\K$.  If several
transposed products are needed, such as for instance computing $T_1
\cdot \ell,\dots,T_L \cdot \ell$ as below, the cost becomes
$O(LD\vm)$, using $D$ as an upper bound on
$\deg(T_1),\dots,\deg(T_L)$. One can actually do better, by computing
inductively and storing the products $X_n^i \cdot \ell$, for
$i=0,\dots,D-1$.  Then, the coefficients of $T_1 \cdot \ell,\dots,T_L
\cdot \ell$ can be computed as the product of the $D \times d'$ matrix
of coefficients of $(X_n^i \cdot \ell)_{0 \le i < D}$ by the matrix of
coefficients of $T_1,\dots,T_L$; the cost is $O(D\vm + LD^2)$.

One can improve this idea further using {\em subproduct tree}
techniques, since the polynomials $T_1,\dots,T_L$ have a very specific
structure.  Recall that we defined $T_k =
P_{\min{}}/P_k^{e_k}$. Hence, all of $T_1,\dots,T_L$ share a common
factor $R=P_{L+1}^{e_{L+1}} \cdots P_{K}^{e_{K}}$.  We can then treat the
common factor $R$ separately, by writing $T_k=R U_k$ for all these
indices $k$, and computing $U_1 \cdot \ell',\dots,U_L \cdot \ell'$
instead, with $\ell'=R \cdot \ell$. The cost to compute $\ell'$ is
$O(D\vm)$.

The polynomials $U_1,\dots,U_L$ have no common factor anymore,
but they are all of the form $P_1^{e_1} \cdots P_{k-1}^{e_{k-1}}
P_{k+1}^{e_{k+1}} P_{L}^{e_{L}}$. We can then define a subproduct tree
as in~\cite[Chapter~10]{GaGe13}, that is, a binary tree $\cal T$ having the
polynomials $(P_k^{e_k})_{1 \le k \le L}$ at its leaves, and where
each node is labeled by the product of the polynomials at its two
children. We proceed in a top-down manner: we associate $\ell'$ to the
root of the tree, and recursively, if a linear form $\lambda$ has been
assigned to an inner node of $\cal T$, we associate to each of its children
the transposed product of $\lambda$ by the polynomial labelling the
other child.  At the leaves, this gives us $U_L \cdot \ell',\dots,U_K
\cdot \ell'$, as claimed. The total cost at each level is $O(D\vm)$,
for a total of $O(D\log(L)\vm)$.

\smallskip\noindent{\rm\bf The main procedure, using the algorithm of Subsection~\ref{ssec:algo1}.}
The first version of the main procedure determines the Gr\"obner bases
of $J_L,\dots,J_K$ by applying the algorithm of
Subsection~\ref{ssec:algo1} to successive families of linear forms.

We maintain a list of ``active'' indices $S$, initially set to
$S=(1,\dots,L)$; these are the indices for which we are not done
yet. The algorithm proceeds iteratively; at step $i\ge 1$, we pick a
random linear form $\ell_i \in Q^*$, and compute all $\ell_{k,i}=T_k
\cdot \ell_i$, for $k$ in $S$. We then apply the algorithm of
Subsection~\ref{ssec:algo1} to
$(\bu_{\ell_{k,1}},\dots,\bu_{\ell_{k,i}})$, for all $k$
independently, and obtain families of polynomials $G_{k,i}$ as
output. For verification purposes, we also choose a random $\ell_0 \in
Q^*$, and compute the corresponding $\ell_{k,0}$.

Write $D_k=\deg(J_k)$, for $k \le K$. Combining
Lemma~\ref{lemma:Dgen} and the equality
$\bu_{\ell_{k,i}}=\bu(\varphi_k(\ell_i))$ seen above, we deduce
that for a generic choice of $\ell_1,\dots,\ell_{D_k}$,
$(\ell_{k,1},\dots,\ell_{k,D_k})$ satisfies assumption ${\sf H}_1$
needed for our algorithm, and that $G_{k,D_k}$ is a Gr\"obner basis of
$J_k$. In view of the discussion in Subsection~\ref{ssec:algo1}, for
any $i < D_k$, $G_{k,i}$ contains a polynomial $g$ not in $J_k$. Since
$\ell_0$ was chosen at random, $\ell_{k,0}$ will in general not vanish
at $g$; hence, at every step $i$, we evaluate $\ell_{k,0}$ at all
elements of $G_{k,i}$, and continue the algorithm for this index $k$
if we obtain a non-zero value; else, we remove $k$ from our list $S$,
and append $G_{k,i}$ to the output.

In terms of complexity, we will have to apply the process in the
previous paragraph to $\mu$ linear forms
$\ell_{D_1},\dots,\ell_{\mu}$, with $\mu=\max_{k \le L}(D_k)$, for a
cost $O(\mu D\vm\log(L))$. Then, we will exploit a feature of
Marinari-M\"oller-Mora's second algorithm: it is incremental in the
number of linear forms given as input, so that the overall runtime of
our $D_k$ successive invocations is the same as if we called it once
with $\ell_1,\dots,\ell_{D_k}$. For a given $k$, it adds up to $O(n
D_k^2 \vm + nD_k^3)=O(n D_k^2 \vm)$, where the first term describes
the cost of the evaluations of the linear forms we need (since 
each new value requires the product by one of the $\mM_i$). Overall, the
runtime is $O(\mu D\log(L)\vm + n \sum_{k \le L} D_k^2 \vm)$.  This
supports the comment made in the introduction: if the degrees of the
multiple components are small, say $D_k = O(1)$ for all $k$, this is
$O(n D \log(D) \vm)$.

\smallskip\noindent{\rm\bf Using the algorithm of Subsection~\ref{ssec:algo2}.}
We can adapt our main procedure in order to use the algorithm of
Subsection~\ref{ssec:algo2} instead; the main difference is that we
expect to use fewer linear forms.

For $k \le K$, let indeed $t_k \le D_k$ be the maximum of
$\tau(Q_{k,\ge 1}),\dots,\tau(Q_{k,\ge n})$, with $Q_{k, \ge
  j}=\K[X_j,\dots,X_n]/J_k \cap\K[X_j,\dots,X_n]$, and with $\tau$
defined as in Proposition~\ref{prop:tau} (for instance, if $I$ is a
complete intersection ideal, $t_k=1$ for all $k$). The main
algorithm proceeds as in the previous variant: we choose random linear
forms $\ell_1,\dots$ and deduce $\ell_{k,i}=T_k \cdot \ell_i$; we will
compute the Gr\"obner basis $G_k$ of $J_k$ as
$\ann(\bu_{\ell_{k,1}},\bu_{\ell_{k,2}},\dots)$.  We claim that we
only need $t_k$ linear forms $\ell_1,\dots,\ell_{t_k}$ in order
to recover $G_k$.

To confirm this, we consider again assumptions ${\sf H}_2$ and ${\sf
  H}_3$ made in Subsection~\ref{ssec:algo2}. The appendix
of~\cite{BoSaSc03} implies that the minimal polynomial of any variable
$X_i$ in $Q_k$ has degree at most $e_k$, except for $X_n$. We already
know the minimal polynomial $P_k^{e_k}$ of $X_n$ in $Q_k$, so we skip
the first pass in the loop of the algorithm of
Subsection~\ref{ssec:algo2}, and use the value $B=e_k$.

Regarding ${\sf H}_3$, we prove that if $\ell_1,\dots,\ell_{t_k}$
are chosen generically, assumption ${\sf H}_3(j)$ holds for
$j=1,\dots,n$. For $i \ge 1$ and $j=1,\dots,n$, define $\ell_{k,i,j}$
as the linear form in $Q^*_{k,\ge j}$ induced by restriction of
$\varphi_k(\ell_{i}) \in Q^*_k$.  Applying Proposition~\ref{prop:tau}
to $Q_{k,\ge j}$ shows that there exists a Zariski open $\Omega_{k,j}
\subset {Q_{k,\ge j}^*}^{t_k}$ such that if
$\ell_{k,1,j},\dots,\ell_{k,t_k,j}$ are in $\Omega_{k,j}$, they
generate $Q_{k,\ge j}^*$ as a $Q_{k,\ge j}$-module, and thus
(Lemma~\ref{lemma:genmod}) $J_k \cap
\K[X_j,\dots,X_n]=\ann(\bu_{\ell_{k,1,j}},\dots,\bu_{\ell_{k,t_k,j}})$. If
this is true for some index $k$ and all $j$, ${\sf H}_3(j)$ follows as
well for these indices. Now, the mapping $\Delta_{k,j}:
(\ell_1,\dots,\ell_{t_k})\mapsto
(\ell_{k,1,j},\dots,\ell_{k,t_k,j})$ is $\K$-linear and onto (we
proved above that $(\ell_1,\dots,\ell_{t_k})\mapsto
(\varphi_k(\ell_{1}),\dots,\varphi_k(\ell_{t_k}))$ is onto, and the
surjectivity of the projection is straightforward), so that the
preimage $\Delta_{k,j}^{-1}(\Omega_{k,j})$ is Zariski open in
${Q^*}^{t_k}$ for all $k,j$. In other words, for generic
$\ell_1,\dots,\ell_{t_k}$, ${\sf H}_3(j)$ holds for all $j$ and all
$k$, so the algorithm of Subsection~\ref{ssec:algo2} computes $G_k$
for all $k$.

We still need to discuss what happens when applying this algorithm to
$\ell_{k,1},\dots,\ell_{k,i}$ for some $i < t_k$. In this case, as
per the discussion in Subsection~\ref{ssec:algo2}, either we get
generators of $\ann(\bu_{\ell_{k,1}},\dots,\bu_{\ell_{k,i}})$, which is a strict
superset of $J_k$, or at least one of the polynomials in the output
does not belong to $\ann(\bu_{\ell_{k,1}},\dots,\bu_{\ell_{k,i}})$. In any case,
the output contains at least one polynomial $g$ not in $J_k$, so we
can use the same stopping criterion as in the previous paragraph,
using a linear form $\ell_0$ to test termination.

To control the complexity, at the $i$th step, we now use linear forms
$\ell_1,\dots,\ell_{2^i}$; as a result, we need to go up to $i=t$,
with $t=\max_k(t_k)$, and the overall runtime is proportional to that
at $i=t$. The cost of preparing the linear forms $\ell_{k,i}$ is $O(t
D\vm\log(L))$, and the cost of computing annihilators is $O(n t
\sum_{k \le L} e_k^2 D_k^2 \vm)$.  The first term is better than the
equivalent term for our first algorithm, but the second one is
obviously worse. On the other hand, the analysis in
Subsection~\ref{ssec:algo2} can be refined significantly, and possibly
lead to improved estimates.

\smallskip\noindent{\rm\bf Using a scalar extension.} To conclude, we 
discuss (without giving proofs) how to put to practice the idea
introduced in Subsection~\ref{sec:output} of computing Gr\"obner bases
of ideals of smaller degree over larger base fields, in the context (for definiteness) of
the algorithm of the previous paragraph.

Let $\ell_{k,1},\dots,\ell_{k,t_k}$ be defined as before, let
$\bu_{\ell_{k,1}},\dots,\bu_{\ell_{k,t_k}}$ be the corresponding
sequences, and assume that these linear forms are such that the
annihilator of $\bu_{\ell_{k,1}},\dots,\bu_{\ell_{k,t_k}}$ is
$J_k$. Let further $\L_k$ be the field extension $\K[Z]/P_n(Z)$, and
let $\zeta_k$ be the residue class of $Z$ in $\L_k$ Then, the
annihilator of $J'_k=J_k + \langle X_n -\zeta_k \rangle^{e_k}$ in
$\L[X_1,\dots,X_n]$ has degree $D_k/f_k$ by Lemma~\ref{lemma:DDprime},
so we might want to compute it instead of $J_k$. To accomplish this,
we need sequences whose annihilator would be $J'_k$; we do this
following the same strategy as above. Define $S_k=P_k/(X_n-\zeta_k)
\in \L_k[X_n]$, as well as the linear form $\ell'_{k,i} = S_k^{e_k}
\cdot \ell_{k,i}: \L[X_1,\dots,X_n]/I\to \L$, for $i \ge 1$. Then,
one verifies that $\ann(\bu_{\ell'_{k,1}},\dots,\bu_{\ell'_{k,t_k}})$
is indeed $J'_k$.

Our last comment discusses the translation mentioned in
Subsection~\ref{sec:output}. The ideal $J'_k$ is $\m'$-primary, with
$\m'=\langle X_1-\xi_1,\dots,X_n-\xi_n\rangle$, as in
Subsection~\ref{sec:output}.  To replace $J'_k$ by a 
$\langle X_1,\dots,X_n\rangle$-primary ideal, we need to modify the
sequences $\bu_{\ell'_{k,1}},\dots,\bu_{\ell'_{k,t_k}}$. For
$i \ge 1$, let $U_{k,i} \in \L[[X_1,\dots,X_n]]$ be the generating series
of $\bu_{\ell'_{k,i}}$, and let $\tilde U_{k,i}= \frac1 {(1+\xi_1
  X_1)\cdots (1+\xi_n X_n)}U_{k,i} (\frac{X_1}{1+\xi_1
  X_1},\dots,\frac{X_n}{1+\xi_nX_n})$. Letting
$\tilde \bu_{{k,i}}$ be the sequence whose generating series is
$\tilde U_{k,i}$, $\ann(\tilde \bu_{{k,1}},\dots,\tilde\bu_{{k,t_k}})$ is
indeed the $\langle X_1,\dots,X_n\rangle$-primary ideal $J_k''$
obtained by translation by $(\xi_1,\dots,\xi_n)$ in $J'_k$.

\bibliographystyle{plain}

\end{document}